\documentclass[graybox]{svmult}

\usepackage{type1cm}        
\usepackage{makeidx}         
\usepackage{graphicx}        
\usepackage{multicol}        
\usepackage[bottom]{footmisc}

\usepackage[english]{babel}
\usepackage{newtxtext}       %
\usepackage[varvw]{newtxmath}       

\makeindex             

\newcommand{\adom}[1]{\mathsf{dom}(#1)}
\newcommand{\aadom}[1]{\mathsf{adom}(#1)}
\newcommand{\ins}[1]{\mathbf{#1}}
\newcommand{\facts}[1]{\mathsf{facts}(#1)}
\newcommand{\class}[1]{\mathsf{#1}}
\newcommand{\dep}{\Sigma}
\newcommand{\ra}{\rightarrow}
\newcommand{\body}[1]{\mathsf{body}(#1)}
\newcommand{\head}[1]{\mathsf{head}(#1)}
\newcommand{\ar}[1]{\mathsf{ar}(#1)}

\begin{document}

\title*{Characterizing Data Dependencies Then and Now\thanks{This is an Author Accepted Manuscript version of the following chapter: Phokion G. Kolaitis and Andreas Pieris, Characterizing Data Dependncies Then and Now, published in Festschrift for Johann A. Makowsky, edited by Klaus Meer, Alexander Rabinovich, Elena Ravve and Andr\'{e}s Villaveces, 2024, Springer Nature Switzerland AG.}}
\author{Phokion G.\ Kolaitis\orcidID{0000-0002-8407-8563} and\\ Andreas Pieris\orcidID{0000-0003-4779-3469}}
\institute{Phokion G.\ Kolaitis \at University of California Santa Cruz \& IBM Research, \email{kolaitis@ucsc.edu}
\and Andreas Pieris \at University of Edinburgh \& University of Cyprus, \email{apieris@inf.ed.ac.uk}}

\maketitle

\abstract*{Data dependencies are integrity constraints that the data of interest must obey. During the 1980s, J\'anos Makowsky made a number of contributions to the study of data dependencies; in particular, he was the first researcher to characterize data dependencies in terms of their structural properties. The goal of this article is to first present an overview of Makowsky's work on characterizing certain classes of data dependencies and then discuss recent developments concerning characterizations of broader classes of data dependencies.}

\abstract{
	Data dependencies are integrity constraints that the data of interest must obey. During the 1980s, J\'anos Makowsky made a number of contributions to the study of data dependencies; in particular, he was the first researcher to characterize data dependencies in terms of their structural properties. The goal of this article is to first present an overview of Makowsky's work on characterizing certain classes of data dependencies and then discuss recent developments concerning characterizations of broader classes of data dependencies.
}

\section{Introduction}\label{sec:introduction}
Since E.F.\ Codd introduced  the relational data model   in 1970 \cite{DBLP:journals/cacm/Codd70}, logic and databases have enjoyed a  continuous and fruitful interaction, so much so that it has been said that ``logic and databases are inextricably intertwined" \cite{Date-logic-db}. There are two main uses of logic in databases: the use of logic as a declarative language to express queries posed on databases and the use of logic as a specification language to express data dependencies, i.e., integrity constraints that the data of interest must  obey. 

Codd \cite{DBLP:persons/Codd72} had the key insight that first-order logic can be used as a database query language, which he called relational calculus. Furthermore, Codd showed that the expressive power of relational calculus (i.e., first-order logic on databases) coincides with that of relational algebra, which is a procedural database query language based on five basic operations on relations (union, difference, cartesian produt, projection, and selection). As regards data dependencies, Codd \cite{codd1972further} introduced the class of functional dependencies, which to date constitute the most widely used such class of constraints. Soon after this, researchers introduced and studied several different classes of data dependencies, such as inclusion dependencies, join dependencies, and multi-valued dependencies. This plethora of types of data dependencies raised the question of identifying a unifying formalism for them. Logic came to the rescue as it was eventually realized that the class of \emph{embedded implicational dependencies} (EIDs)  provides such a formalism \cite{DBLP:journals/jacm/BeeriV84,DBLP:journals/jacm/Fagin82,DBLP:conf/focs/YannakakisP80}. Embedded implicational dependencies comprise two  classes of first-order sentences, the class of tuple-generating dependencies (tgds) and the class of equality-generating dependencies (egds); the latter  generalizes functional dependencies, while the former generalizes, among others, inclusion, join, and multi-valued dependencies.  Informally, a tgd asserts that if some tuples belong to some relations, then some other tuples must belong to some (perhaps different) relations, while an egd asserts that if some tuples belong to some relations, then two of the values occurring  in some of these tuples must be equal. More precisely, a tgd is a universal-existential first-order sentence of the form
\[
\forall \bar x\forall \bar y \, (\phi(\bar x, \bar y)\ \rightarrow\ \exists \bar z\psi(\bar x, \bar z)),
\]
where $\phi(\bar x,\bar y)$ and $\psi(\bar x, \bar z)$ are conjunctions of atomic formulas. Furthermore, an egd is a universal first-order sentence of the form
\[
\forall \bar x \, (\phi(\bar x)\ \rightarrow\ x_i=x_j),
\]
where $\phi(\bar x)$ is a conjunction of atomic formulas, and $x_i$, $x_j$ are  variables in $\bar x$.

During the 1980s, J\'anos Makowsky made remarkable contributions to the study of data dependencies. The first of these contributions concerns the 
implication problem for data dependencies, which by the late 1970s had emerged as the central problem in this area and, in fact, had  been called \emph{the fundamental problem of databases} (see~\cite{DBLP:conf/csl/Makowsky11}). This problem takes as input a finite set $\Sigma$ of data dependencies and a data dependency $\sigma$, and asks whether $\Sigma$ logically implies $\sigma$. Here, logical implication has two different versions: the first version is that $\sigma$ is true on every database (finite or infinite) on which every member of $\Sigma$ is true; the second version is that this implication holds in the finite, i.e., $\sigma$ is true on every finite database on which every member of $\Sigma$ is true. In the case of functional dependencies, it was known that the two versions of the implication problem coincide and that there is a polynomial-time algorithm for solving this problem. It was not known, however, whether the implication problem for EIDs was decidable or undecidable. Makowsky showed that both versions of the implication problem for EIDs are undecidable.  This result appeared in a joint paper with Chandra and Lewis \cite{DBLP:conf/stoc/ChandraLM81}, who had independently arrived at the same solution. The proof was via a reduction from the halting problem for two-counter machines. A different reduction that uses  undecidability results from equational logic was obtained by Beeri and Vardi \cite{DBLP:conf/icalp/BeeriV81} around the same time. Chandra, Lewis, and Makowsky \cite{DBLP:conf/stoc/ChandraLM81} also showed that the implication problem is decidable for \emph{full} tgds and egds, where a tgd is full if it has no existential quantifiers, i.e., it is a first-order sentence of the form
\[
\forall  \bar x \forall \bar y \, (\phi(\bar x,\bar y)\ \rightarrow\ \psi(\bar x)),
\]
where $\phi(\bar x,\bar y)$ and $\psi(\bar x)$ are conjunctions of atomic formulas.

Makowsky's second contribution to the study of data dependencies concerns   characterizations 
of classes of data dependencies. To put Makowsky's work in a proper context, let us recall three important topics  in mathematical logic. First, in his invited address at the 1950 International Congress of Mathematicians, Tarski articulated his interest in characterizing notions of ``metamathematical origin'' in ``purely mathematical terms''~\cite{tarski1950some}. A series of results in model theory eventually led to the following characterization of  definability in first-order logic: a class $C$ of structures is definable by a finite set of first-order sentences if and only if both the class $C$ and its complement $\overline{C}$ are closed under isomorphisms and ultraproducts (see \cite{DBLP:books/daglib/0067423}). Thus, the ``metamathematical'' notion of definability by a finite set of first-order sentences can be characterized in terms of purely algebraic closure properties. 
Second, there is a body of results in model theory, known as \emph{preservation theorems}, that characterize when a first-order sentence is logically equivalent to a first-order sentence of a restricted syntactic form. The prototypical preservation theorem is the {\L}o\'s-Tarski Theorem, which asserts that a first-order sentence is logically equivalent to a universal first-order sentence if and only if it is preserved under substructures. Third,  Lindstr\"om~\cite{lindstrom1969extensions} characterized first-order logic as a maximal logic that has some mild closure properties and satisfies the compactness theorem and the Skolem-L\"owenheim Theorem. This intrinsic characterization of first-order logic became the catalyst for  the development of abstract model theory, which aims at characterizing logical formalisms in terms of their properties.

Leveraging his expertise in mathematical logic, Makowsky worked on characterizations of classes of data dependencies.  In~\cite{DBLP:conf/icalp/Makowsky81}, he focused  on full tgds and egds over uni-relational databases, that is, all data dependencies considered were over a database schema consisting of a single relation symbol.  Note that much (but not all) of the work on data dependencies at that time was about uni-relational databases, perhaps because the first data dependencies studied were the functional dependencies, and they involve a single relation symbol. Makowsky obtained a preservation theorem that characterizes when a first-order sentence is logically equivalent to a finite set of full tgds and egds. This preservation theorem entails closure under subdatabases and closure under direct products. It is worth pointing out that, while this result is aligned with Tarski's goal to characterize metamathematical notions in purely mathematical terms, Makowsky questioned its relevance to databases because, as he put it, ``algebraic operations alone is not what is needed in data base theory''. In other words, Makowsky argued that characterizations of data dependencies should involve notions that are more ubiquitous in database practice. With these considerations in mind, Makowsky went on   to establish the main result in~\cite{DBLP:conf/icalp/Makowsky81}, which is a Lindstr\"om-type theorem for the class of all full tgds and egds.  Specifically, he showed that the class of all full tgds and egds is a maximal collection of first-order sentences that satisfy a locality condition, called \emph{securability}, and admit \emph{Armstrong relations}, which are relations that encapsulate precisely the full tgds and egds that are logically implied by a set of full tgds and egds.

After this, Makowsky and Vardi \cite{DBLP:journals/acta/MakowskyV86} studied data dependencies over multi-relational databases, that is, data dependencies over a database schema that consists of finitely many relation symbols. Makowsky and Vardi \cite{DBLP:journals/acta/MakowskyV86} characterized when a class of databases is axiomatizable by a set of full tgds and egds,   and also  when it is axiomatizable by a finite set of full tgds and egds. These characterizations involve closure under subdatabases, closure under direct products, suitable locality properties, and some other structural properties that will be discussed in the sequel.

During the past two decades, data dependencies have found new uses and applications in several different areas of data management and knowledge representation, including data exchange, data integration, and specification of ontologies. As a result, the interest in characterizing data dependencies has been rekindled~\cite{DBLP:conf/pods/ConsoleKP21,DBLP:conf/icdt/CateK09}; furthermore,  new preservation theorems about data dependencies have been obtained \cite{zhang2022characterizing,ZhZJ20}. In particular, using a novel notion of locality, characterizations of arbitrary tgds and egds were obtained in \cite{DBLP:conf/pods/ConsoleKP21}, thus going well beyond the characterizations of full tgds and egds established in~\cite{DBLP:journals/acta/MakowskyV86}. 
In this article, we present some of these developments in detail and discuss the motivating role played by the concepts and results in~\cite{DBLP:journals/acta/MakowskyV86}.


\section{Relational Databases and Data Dependencies}\label{sec:preliminaries}

Let $\ins{C}$ and $\ins{V}$ be disjoint countably infinite sets of constants and variables, respectively. For an integer $n > 0$, we may write $[n]$ for the set $\{1,\ldots,n\}$.

\medskip

\noindent
\textbf{Relational Databases.} A {\em (relational) schema} $\ins{S}$ is a finite set of relation symbols (or predicates) with positive arity; let $\ar{R}$ be the arity of $R$.
A {\em (relational) database} $D$ over $\ins{S} = \{R_1,\ldots,R_n\}$, or simply {\em $\ins{S}$-database}, is a tuple $(\adom{D},R_{1}^{D},\ldots,R_{n}^{D})$, where $\adom{D} \subseteq \ins{C}$ is a finite domain and $R_{1}^{D},\ldots,R_{n}^{D}$ are relations over $\adom{D}$, i.e., $R_{i}^{D} \subseteq \adom{D}^{\ar{R_i}}$ for $i \in [n]$. In logical terms, a database is simply a relational structure with finite domain. A {\em fact} of $D$ is an expression of the form $R_i(\bar c)$, where $\bar c \in R_{i}^{D}$. Let $\mathsf{facts}(D)$ be the set of facts of $D$.
The {\em active domain} of $D$, denoted $\aadom{D}$, is the set of elements of $\adom{D}$ that occur in at least one fact of $D$. 
An $\ins{S}$-database $D'$ is a {\em subdatabase} of $D$, denoted $D' \preceq D$, if $\adom{D'} \subseteq \adom{D}$ and, for each $R \in \ins{S}$, we have that  $R^{D'} = R_{\mid \adom{D'}}^{D}$, where $R_{\mid \adom{D'}}^{D}$ is  the {\em restriction of $R^{D}$ over $\adom{D'}$}, i.e., $R_{\mid \adom{D'}}^{D} = \left\{\bar c \in R^{D} \mid \bar c \in \adom{D'}^{\ar{R}}\right\}$.
A {\em homomorphism} from an $\ins{S}$-database $D$ to an $\ins{S}$-database $D'$ is a function $h : \adom{D} \ra \adom{D'}$ with the following property:  for each $i \in [n]$ and for each tuple $(c_1,\ldots,c_m) \in R_{i}^{D}$, we have that  $(h(c_1),\ldots,h(c_m)) \in R_{i}^{D'}$. We write $h : D \ra D'$ to denote that $h$ is a homomorphism from $D$ to $D'$. We also write $h(\facts{D})$ for the set $\{R(h(\bar c)) \mid R(\bar c) \in \facts{D}\}$. Finally, we say that $D$ and $D'$ are {\em isomorphic}, written $D \simeq D'$, if there is a bijection $h : \adom{D} \ra \adom{D'}$ such that $h$ is a homomorphism from $D$ to $D'$ and $h^{-1}$ is a homomorphism from $D'$ to $D$.

\medskip
\noindent \textbf{Tuple-Generating Dependencies.} An atom over $\ins{S}$ is an expression of the form $R(\bar v)$, where $R \in \ins{S}$ and $\bar v$ is an $\ar{R}$-tuple of variables from $\ins{V}$.
A {\em tuple-generating dependency} (tgd) $\sigma$ over a schema $\ins{S}$ is a constant-free first-order sentence
\[
\forall \bar x \forall \bar y \left(\phi(\bar x,\bar y)\ \ra\ \exists \bar z\, \psi(\bar x,\bar z)\right),
\]
where $\bar x, \bar y, \bar z$ are tuples of variables of $\ins{V}$, $\phi(\bar x,\bar y)$ is a (possibly empty) conjunction of atoms over $\ins{S}$, and $\psi(\bar x,\bar z)$ is a non-empty conjunction of atoms over $\ins{S}$.
For brevity, we write $\sigma$ as $\phi(\bar x,\bar y) \ra \exists \bar z\, \psi(\bar x,\bar z)$. We refer to $\phi(\bar x,\bar y)$ and $\psi(\bar x,\bar z)$ as the {\em body} and {\em head} of $\sigma$, denoted $\body{\sigma}$ and $\head{\sigma}$, respectively.
By an abuse of notation, we may treat a tuple of variables as a set of variables, and we may also treat a conjunction of atoms as a set of atoms.
A tgd is called {\em full} if it has no existentially quantified variables. Let us stress that, since the head of a tgd is by definition non-empty, full tgds mention at least one universally quantified variable, which in turn implies that their body is non-empty.
An $\ins{S}$-database $D$ \emph{satisfies a tgd} $\sigma$ as the one above, written $D \models \sigma$, if the following holds: whenever there is a function $h : \bar x \cup \bar y \ra \adom{D}$ such that $h(\phi(\bar x, \bar y)) \subseteq \facts{D}$ (as usual, we write $h(\phi(\bar x,\bar y))$ for the set $\{R(h(\bar v)) \mid R(\bar v) \in \phi(\bar x,\bar y)\}$), then there exists an extension $h'$ of $h$ such that $h'(\psi(\bar x, \bar z)) \subseteq \facts{D}$. The $\ins{S}$-database $D$ \emph{satisfies a set $\dep$ of tgds}, written $D \models \dep$, if $D \models \sigma$ for each $\sigma \in \dep$; in this case, we say that $D$ is a {\em model} of $\dep$.

\medskip
\noindent \textbf{Equality-Generating Dependencies.} An {\em equality-generating dependency} (egd) $\theta$ over a schema $\ins{S}$ is a constant-free first-order sentence
\[
\forall \bar x \left(\phi(\bar x)\ \ra\ x_i = x_j\right),
\]
where $\bar x$ is a tuple of variables of $\ins{V}$, $\phi(\bar x)$ is a (non-empty) conjunction of atoms over $\ins{S}$, and $x_i,x_j \in \bar x$.
For brevity, we write $\theta$ as $\phi(\bar x) \ra x_i = x_j$. We refer to $\phi(\bar x)$ as the {\em body} of $\theta$, denoted $\body{\theta}$.
An $\ins{S}$-database $D$ \emph{satisfies an egd} $\theta$ as the one above, written $D \models \theta$, if, whenever there is a function $h : \bar x \ra \adom{D}$ such that $h(\phi(\bar x)) \subseteq \facts{D}$, then $h(x_i) = h(x_j)$. The $\ins{S}$-database $D$ \emph{satisfies a set $\dep$ of egds}, written $D \models \dep$, if $D \models \theta$ for each egd $\theta \in \dep$, and  we say that
$D$ is a {\em model} of $\dep$.

\medskip
\noindent \textbf{Finite Axiomatizability.} Let $C$ be a collection of databases. We say that $C$ is {\em finitely axiomatizable} by tgds and egds if there are finite sets $\dep_T$ and $\dep_E$ of tgds and egds, respectively, such that for every database $D$, it holds that $D \in C$ if and only if $D \models \dep_T$ and $D \models \dep_E$. Note that in the literature we also have the notion of axiomatizability, where the sets $\dep_T$ and $\dep_E$ can be infinite, i.e., if we allow $\dep_T$ and $\dep_E$ to be infinite, then we say that $C$ is {\em axiomatizable} by tgds and egds. 
In this paper, we focus on finite axiomatizability since it is more relevant to database practice. In what follows, we consider only collections $C$ of databases that are {\em closed under isomorphisms}, i.e., if $D \in C$ and $D'$ is a database such that $D \simeq D'$, then $D' \in C$. In other words, for a collection of databases, we silently assume that it is closed under isomorphisms.


\section{Finite Axiomatizability by Full TGDs and EGDs}\label{sec:ftgds-egds}

As discussed in Section~\ref{sec:introduction}, one of Makowsky's main contributions to the study of data dependencies concerns characterizations of classes of data dependencies. The result that stands out in this context, which we will present here, is the characterization of finite sets of {\em full dependencies} obtained by Makowsky and Vardi in 1986~\cite[Theorem 5]{DBLP:journals/acta/MakowskyV86}. By definition, a full dependency is a first-order sentence of the form $\forall \bar x \forall \bar y\, (\phi(\bar x,\bar y) \ra \psi(\bar x))$, where $\phi(\bar x,\bar y)$ is a non-empty conjunction of relational atoms and $\psi(\bar x)$ is a non-empty conjunction of relational atoms and equality atoms. Clearly, every full dependency is logically equivalent to a finite set of full tgds and egds. Consequently, given a full dependency $\sigma$, there is a finite set $\dep$ of full tgds and egds such that, for every database $D$, it holds that $D$ satisfies $\sigma$ if and only if $D$ satisfies the tgds and egds of $\dep$. For this reason, in what follows we will state Theorem~5 of~\cite{DBLP:journals/acta/MakowskyV86} in terms of finite sets of full tgds and egds.

As we shall see, Theorem 5 of \cite{DBLP:journals/acta/MakowskyV86} involves the property of closure under direct products, which is formally defined below. However, as already mentioned in Section~\ref{sec:introduction}, Makowsky questioned the relevance to databases of closure under algebraic operations, such as closure under direct products. With this in mind, we further present an alternative characterization of finite sets of full tgds and egds, which uses the property of closure under intersections instead of closure under direct products.

\subsection{The Characterization by Makowsky and Vardi}

We first introduce the properties that are needed for the characterization of interest and then recall the characterization by Makowski and Vardi from~\cite{DBLP:journals/acta/MakowskyV86}.

\subsubsection*{Model-theoretic Properties}

In the sequel, fix an arbitrary schema $\ins{S} = \{R_1,\ldots,R_\ell\}$.

\medskip
\noindent \textbf{$1$-Criticality.} An $\ins{S}$-database $D = (\adom{D},R_{1}^{D},\ldots,R_{\ell}^{D})$ is said to be {\em $1$-critical} if $|\adom{D}| = 1$, that is, the domain of $D$ consists of a single element $c \in \ins{C}$, and $R_{i}^{D} = \{c\}^{\ar{R_i}}$ for each $i \in [\ell]$, i.e., $\facts{D} = \{R_i(c,\ldots,c) \mid i \in [\ell]\}$. Let us clarify that in~\cite{DBLP:journals/acta/MakowskyV86} $1$-critical databases are called {\em trivial} databases.
A collection of databases over $\ins{S}$ is {\em $1$-critical} if it contains an $1$-{\rm critical} $\ins{S}$-database.

\medskip
\noindent \textbf{Domain Independence.}
A collection $C$ of databases over $\ins{S}$ is {\em domain independent} if, for every $\ins{S}$-database $D \in C$ and every database $D'$ with 
$\facts{D} = \facts{D'}$, it holds that 
$D' \in C$. In other words, $C$ is domain independent if, for every two $\ins{S}$-databases that have the same set of facts, but not necessarily the same domain, either both are in $C$ or neither of them is in $C$.

\medskip
\noindent \textbf{Modularity.}
A collection $C$ of databases over $\ins{S}$ is {\em $n$-modular}, for $n \geq 0$, if, for every $\ins{S}$-database $D \not\in C$, there is an $\ins{S}$-database $D' \preceq D$ with $|\adom{D'}| \leq n$ such that $D'\not\in C$.
We say that $C$ is {\em modular} if it is $n$-modular for some $n \geq 0$.
Roughly, $n$-modularity provides a ``small'' database with at most $n$ domain elements as a witness to why a database does not belong to a collection of databases.
Let us clarify that in~\cite{DBLP:journals/acta/MakowskyV86}, $n$-modularity is called $n$-{\em locality}. However, here we adopt the term $n$-modularity, which has been already used in~\cite{DBLP:conf/icdt/CateK09} for a similar notion in the context of schema mappings, in order to avoid any confusion with the notion of $(n,m)$-locality from~\cite{DBLP:conf/pods/ConsoleKP21}, which will be used in the next section to handle arbitrary tgds and egds.

\medskip
\noindent \textbf{Closure Under Subdatabases.} A collection $C$ of databases is {\em closed under subdatabases} if, for every $\ins{S}$-database $D \in C$, it holds that $D' \in C$ for every $D' \preceq D$.

\medskip
\noindent \textbf{Closure Under Direct Products.} 
Assume that    $D = (\adom{D},R_{1}^{D},\ldots,R_{\ell}^{D})$ and $D' = (\adom{D'},R_{1}^{D'},\ldots,R_{\ell}^{D'})$ are two databases over a schema $\ins{S}$. The {\em direct product} of $D$ and $D'$, denoted $D \otimes D'$, is defined as the $\ins{S}$-database 
$$D'' = (\adom{D''},R_{1}^{D''},\ldots,R_{\ell}^{D''}),$$ where $\adom{D''} = \adom{D} \times \adom{D'}$, and, for each $i \in [\ell]$, $R_{i}^{D''}$ is the relation
\[ \left\{\left((a_1,b_1),\ldots,(a_{r_i},b_{r_i})\right) \mid \left(a_1,\ldots,a_{r_i}\right) \in R_{i}^{D} \text{ and } \left(b_1,\ldots,b_{r_i}\right) \in R_{i}^{D'}\right\},
\]
where $r_i=\ar{R_i}$ is the arity of the relation symbol $R_i$.
A collection $C$ of databases over $\ins{S}$ is {\em closed under direct products} if, for every two $\ins{S}$-databases $D,D' \in C$, it holds that $D \otimes D' \in C$.

\subsubsection*{The Characterization}

We can now provide the characterization of interest from~\cite{DBLP:journals/acta/MakowskyV86}. We do not discuss here the proof of the characterization since it will be derived from the more general characterization of finite sets of arbitrary tgds and egds presented in Section~\ref{sec:tgds-egds}.

\begin{theorem}[Theorem 5 in~\cite{DBLP:journals/acta/MakowskyV86}]\label{the:tgd+egd-MV86}
	Let $C$ be a collection of databases. The following statements are equivalent:
	\begin{enumerate}
		\item $C$ is finitely axiomatizable by full tgds and egds.
		\item $C$ is $1$-critical, domain independent, modular, closed under subdatabases, and closed under direct products.
	\end{enumerate}
\end{theorem}

\subsection{The Alternative Characterization}

We will present an alternative characterization of finite sets of full tgds and egds that avoids the use of closure under direct products. This is done by replacing in the characterization of Theorem~\ref{the:tgd+egd-MV86} the property of closure under direct products with the property of closure under intersections defined as follows.
Consider a schema $\ins{S} = \{R_1,\ldots,R_\ell\}$. The {\em intersection} of the $\ins{S}$-databases $D = (\adom{D},R_{1}^{D},\ldots,R_{\ell}^{D})$ and $D' = (\adom{D'},R_{1}^{D'},\ldots,R_{\ell}^{D'})$, denoted $D \cap D'$, is the database 
\[
(\adom{D} \cap \adom{D'},R_{1}^{D} \cap R_{1}^{D'},\ldots,R_{\ell}^{D} \cap R_{\ell}^{D'}). 
\]
A collection $C$ of databases over $\ins{S}$ is {\em closed under intersections} if, for every two $\ins{S}$-databases $D,D' \in C$, it holds that $D \cap D' \in C$.

Before stating and proving the promised characterization, let us briefly discuss the two tools that will be used in the proof: Robinson's \emph{method of diagrams} and McKinsey's \emph{method of eliminating disjunctions}.
In his Ph.D. thesis in 1949, Abraham Robinson introduced the notion of a \emph{diagram} of a relational structure, which, together with its variants, became a standard tool in constructing models that contain an isomorphic copy of some  structure of interest. In model theory, the method of  diagrams is typically combined with the compactness theorem to construct such models (see~\cite{DBLP:books/daglib/0067423,robinson63}). In characterizing data dependencies, however, the method of diagrams is combined with closure properties and with suitable notions of locality, instead of the compactness theorem. In 1943, J.C.C.\ McKinsey \cite{mckinsey1943decision} introduced the method of eliminating disjunctions, which makes it possible to show that, when a class  of structures satisfies certain closure properties, then, on such a class of structures, a first-order sentence of the form
$\forall \bar x (\phi(\bar x)\rightarrow \bigvee_{i=1}^k\psi_i(\bar x))$   implies one of the sentences
$\forall \bar x (\phi(\bar x)\rightarrow \psi_j(\bar x))$, for some $j \in [k]$; as a consequence of this, on such a class  of structures, the sentence   $\forall \bar x (\phi(\bar x)\rightarrow \bigvee_{i=1}^k\psi_i(\bar x))$  is equivalent to the sentence 
$\forall \bar x (\phi(\bar x)\rightarrow \psi_j(\bar x))$. McKinsey's original application of this method used closure under direct products, but the method is flexible enough to adapt to other closure properties, such as closure under intersections. It should be noted that the method of diagrams and the method of eliminating disjunctions were already utilized by Makowsky and Vardi \cite{DBLP:journals/acta/MakowskyV86} in the proof of Theorem \ref{the:tgd+egd-MV86}.

\begin{theorem} \label{the:ftgd+egd-intersections}
	Let $C$ be a collection of databases. The following are equivalent:
	\begin{enumerate}
		\item $C$ is finitely axiomatizable by full tgds and egds.
		\item $C$ is $1$-critical, domain independent, modular, closed under subdatabases, and closed under intersections.
	\end{enumerate}
\end{theorem}

\begin{proof}
	The direction $(1)\Rightarrow (2)$ is easy and we leave it as an exercise.
	We focus on the direction
	$(2) \Rightarrow (1)$. We first need to introduce a couple of auxiliary notions. 
	
	A {\em disjunctive dependency} (dd) $\delta$ over a schema $\ins{S}$ is a constant-free sentence
	\[
	\forall \bar x \left(\phi(\bar x)\ \ra\ \bigvee_{i=1}^{k}
	\psi_i(\bar x_i)\right),
	\]
	where $\bar x$ is a tuple of variables of $\ins{V}$, the expression $\phi(\bar x)$ is a (non-empty) conjunction of atoms over $\ins{S}$, and, for each $i \in [k]$, $\bar x_i \subseteq \bar x$ and the expression $\psi_i(\bar x_i)$ is either an equality formula $y=z$ with $\bar x_i = \{y,z\}$, or an atom over $\ins{S}$.
	An $\ins{S}$-database $D$ \emph{satisfies the dd} $\delta$, written $D \models \delta$, if, whenever there is a function $h : \bar x \ra \adom{D}$ with $h(\phi(\bar x)) \subseteq \facts{D}$, then there is $i \in [k]$ such that, if $\psi_i(\bar y_i)$ is $y=z$, then $h(y) = h(z)$, and if $\psi_i(\bar y_i)$ is $R(\bar x_i)$, then $h(R(\bar x_i)) \in \facts{D}$. We say that the database $D$ \emph{satisfies a set} $\dep$ of dds or that  $D$ is a {\em model} of $\dep$, written $D \models \dep$, if $D \models \delta$ for each $\delta \in \dep$.

	We further need the notion of the diagram of a database. Let $D$ be an $\ins{S}$-database such that $\facts{D} \neq \emptyset$. Let $A$ be the set of atomic formulas that can be formed using predicates from $\ins{S}$ and constants from $\adom{D}$. The {\em diagram of $D$}, denoted $\Delta_D$, is 
	\[
	\bigwedge_{\alpha \in \facts{D}} \alpha\ \wedge\ \bigwedge_{\alpha \in A \setminus \facts{D}} \neg \alpha\ \wedge\ \bigwedge_{c,d \in \adom{D} \text{ and } c \neq d} \neg (c=d).
	\]
	Note that, since $\facts{D}\neq \emptyset$, the conjunction $\bigwedge_{\alpha \in \facts{D}} \alpha$ is non-empty.

	Let $C$ be collection of databases over $\ins{S}$  possessing the properties in the second statement of Theorem \ref{the:ftgd+egd-intersections}. In particular, since $C$ is modular, there is some integer $n \geq 0$ such that $C$ is $n$-modular. Moreover, since $C$ is closed under subdatabases, we have that $C$ contains the empty database, hence $n$ must be a positive integer.

 Let $\dep^{\vee}$ be the set of all dds over $\ins{S}$ with at most $n$ variables that are satisfied by every database $D \in C$. It is clear that $\dep^\vee$ is finite (up to logical equivalence). We proceed to show the following technical lemma.

	\begin{lemma}\label{lem:aux-lemma-1-full-instersection}
		For each $\ins{S}$-database $D$, we have that $D \in C$ if and only $D \models \dep^\vee$.
	\end{lemma}
	
	\begin{proof}[Proof of Lemma~\ref{lem:aux-lemma-1-full-instersection}]
		Fix an $\ins{S}$-database $D$. From the definition of $\dep^{\vee}$, $D \in C$ implies $D \models \dep^{\vee}$. 
		For the other direction, we will show that if $D \not \in C$, then $D \not \models \dep^{\vee}$. Assume that $D \not\in C$.  Since $C$ is $n$-modular, there exists an $\ins{S}$-database $D_n$ such that $D_n \preceq D$, $|\adom{D_n}| \leq n$, and $D_n \not\in C$. 
		Since $C$ is domain independent, we can assume that $\adom{D_n} = \aadom{D_n}$.
		We will show that $D_n \not\models \dep^\vee$, which in turn implies that $D \not\models \dep^\vee$ since every sentence in $\dep^\vee$ is universal and it is well-known that universal first-order sentences are preserved under subdatabases.  
		
		To show that $D_n \not\models \dep^\vee$, it suffices to find a dd $\delta \in \dep^\vee$ such that $D_n \not \models \delta$. We will use the diagram of $D_n$ to construct this desired $\delta$.
		Let $\Delta_{D_n}$ be the diagram of $D_n$, and let
		$\Phi_{D_n}(\bar x)$ be the quantifier-free first-order formula obtained from $\Delta_{D_n}$ by replacing each $c \in \adom{D_n}$ with a new variable $x_c \in \ins{V}$. We claim that the following statements are true:
		\begin{enumerate}
			\item $D_n \models \exists \bar x \, \Phi_{D_n}(\bar x)$.
			\item For every $D' \in C$, we have that $D' \models \neg \exists \bar x \, \Phi_{D_n}(\bar x)$.
			\item The sentence $\neg \exists \bar x \, \Phi_{D_n}(\bar x)$ is logically equivalent to a dd $\delta$.  
		\end{enumerate}
		The first statement is obviously true by the construction of the sentence $\exists \bar x \, \Phi_{D_n}(\bar x)$ (each existential quantifier $\exists x_c$ in $\exists \bar x$ can be witnessed by the element $c$ in $\adom{D_n}$). 
		For the second statement, if there is some database $D'\in C$ such that $D' \models \exists \bar x \, \Phi_{D_n}(\bar x)$, then $D'$ must contain a subdatabase $D''$ that is isomorphic to $D_n$. Since $C$ is closed under subdatabases and isomorphisms, it follows that $D_n \in C$, which is a contradiction. 
		For the third statement, observe first that $\facts{D_n} \neq \emptyset$ because $D_n$ is not in $C$, while $C$, being closed under subdatabases, contains the empty database. This implies that $\Delta_{D_n}$ contains at least one positive relational atom $\alpha$; hence, the conjunction $\bigwedge_{\alpha \in \facts{D_n}} \alpha$ is non-empty. Furthermore, since $D_n$ is not in $C$ and since $C$ is $1$-critical and closed under isomorphisms, we conclude that $D_n$ is not an $1$-critical database. This implies that either $|\adom{D_n}| \geq 2$ or $|\adom{D_n}| = 1$ and there is a relation symbol $R \in \ins{S}$ such that $R^{D_n}$ is empty. Therefore, $\Delta_{D_n}$ contains either the negation of an equality atom $\beta$ or the negation of a relational atom $\gamma$. Consequently, $\neg \exists \bar x \, \Phi_{D_n}(\bar x)$ is logically equivalent to a dd $\delta$. Therefore, all three statements have been established.
		
		Taken together,
		these three statements imply that $\delta$ is a dd in $\dep^\vee$ that is false on $D_n$, and hence, $D_n\not \models \dep^\vee$. This completes the proof of Lemma \ref{lem:aux-lemma-1-full-instersection}. 
	\end{proof}

	Observe that the property of closure under intersections was not used in the proof of Lemma \ref{lem:aux-lemma-1-full-instersection}. This property will be used in the proof of the next technical lemma.

	\begin{lemma}\label{lem:aux-lemma-2-full-instersection}
		Assume that 
$
		\forall \bar x (\phi(\bar x) \ra
		\bigvee_{i=1}^{k} \alpha_i(\bar x_i))$ 
  is a disjunctive dependency that belongs to $\dep^\vee$.
		There is $j \in [k]$ such that $\forall \bar x (\phi(\bar x) \ra
		\alpha_j(\bar x_j))$ also belongs to $\dep^\vee$.	
	\end{lemma}
	\begin{proof}[Proof of Lemma~\ref{lem:aux-lemma-2-full-instersection}]
 Note that for every
 $j\in [k]$, the sentence
$\forall \bar x (\phi(\bar x) \ra
		\alpha_j(\bar x_j))$ is either a full tgd or an egd, hence it is a disjunctive dependency.
  Towards a contradiction, assume that none of these sentences belongs to $\dep^\vee$.
  Since $\dep^\vee$ consists of all disjunctive dependencies that are satisfied by every database in $C$, it follows that for every $j\in [k]$,
  there is a database
		$D_j \in C$ such that $D_j \not \models \forall \bar x (\phi(\bar x) \ra
		\alpha_j(\bar x_j))$. Hence, for every $j \in [k]$, there is a tuple $\bar b^j$ of elements in $\adom{D_j}$ such that $D_j \models \phi(\bar b^j)\land \neg \alpha_j(\bar b^j_j)$. Since $C$ is closed under isomorphisms, we may assume that these tuples are the same tuple $\bar b$, i.e., we may assume that for every $j$ and $\ell$ with $j,\ell \in [k]$, we have that ${\bar b^j}={\bar b^\ell}=\bar b$.  
		Let $D_{\cap}$ be the intersection of the databases $D_j$, that is, $D_{\cap} = \bigcap_{j\in [k]} D_j$. Since $C$ is closed under intersections, we have that $D_{\cap}$ belongs to $C$;   consequently, we have that $D_{\cap} \models \forall \bar x (\phi(\bar x) \ra \bigvee_{i=1}^{k} \alpha_i(\bar x_i))$. 
		Since the elements of the tuple $\bar b$ belong to $\adom{D_j}$ for every $j \in [k]$, it follows that they also belong to $\adom{D_{\cap}}$. Therefore, and since also $\phi(\bar b)$ holds, there is some $j \in [k]$ such that $\alpha_j(\bar b_j)$ holds, which contradicts the fact that $\neg \alpha_j(\bar b_j)$ holds (since $\bar b$ witnesses that $D_j\not \models \forall \bar x (\phi(\bar x) \ra \alpha_j(\bar x_j))$. This completes the proof that  there is a $j \in [k]$ such that the formula $\forall \bar x (\phi(\bar x) \ra
		\alpha_j(\bar x_i))$ belongs to $\dep^\vee$.
\end{proof}

  With the two preceding lemmas at hand, we are now ready to complete the proof of Theorem \ref{the:ftgd+egd-intersections}.  
    Let $\dep$  be the set of full tgds and egds in $\Sigma^\vee$, i.e.,
		\[\dep = 
		\left\{\delta \in \dep^\vee \mid \delta \text{ is a full tgd or an egd} \right\}.
		\]
Note that $\dep$ is a finite set because it is a subset of the finite set $\dep^\vee$. 
%
We proceed to show that, for each $\bf S$-database $D$, it holds that $D\in C$ if and only if $D\models \dep$.

\medskip

($\Rightarrow$) If $D\in C$, then, by the definition of $\dep^\vee$, $D\models \dep^\vee$; hence, $D \models \dep$ since $\dep \subseteq \dep^\vee$. 

($\Leftarrow$) Assume now that $D$ is an $\ins{S}$-database such that $D\models \dep$. We need to show that $D\in C$. By Lemma \ref{lem:aux-lemma-1-full-instersection}, it suffices to show that $D\models \dep^\vee$. Let 
$\forall \bar x (\phi(\bar x) \ra \bigvee_{i=1}^{k} \alpha_i(\bar x_i))$ be a disjunctive dependency in $\dep^\vee$. By Lemma \ref{lem:aux-lemma-2-full-instersection}, there is $j\in [k]$ such that $\forall \bar x (\phi(\bar x) \ra \alpha_j(\bar x_j))$ belongs to $\dep^\vee$.  Since this  sentence is either a full tgd or an egd, we have that it belongs to $\dep$; hence, $D \models  \forall \bar x (\phi(\bar x) \ra \alpha_j(\bar x_j))$. It is obvious that $\forall \bar x (\phi(\bar x) \ra \alpha_j(\bar x_j))$ logically implies the disjunctive dependency $\forall \bar x (\phi(\bar x) \ra \bigvee_{i=1}^{k} \alpha_i(\bar x_i))$, and therefore, $D\models \forall \bar x (\phi(\bar x) \ra \bigvee_{i=1}^{k} \alpha_i(\bar x_i))$. This completes the proof of the above claim.

\medskip

Thus, we have established that $C$ is axiomatizable by the set $\dep$, which is a finite set of full tgds and egds. This completes the proof of Theorem~\ref{the:ftgd+egd-intersections}.
\end{proof}


\section{Finite Axiomatizability by TGDs and EGDs}\label{sec:tgds-egds}

The fact that Theorems~\ref{the:tgd+egd-MV86} and~\ref{the:ftgd+egd-intersections} deal only with full tgds, i.e., tgds without existentially quantified variables, leads to the following natural question: is there an analogous characterization for arbitrary tgds and egds? Such a characterization has been recently established in~\cite{DBLP:conf/pods/ConsoleKP21} for the case of arbitrary (finite or infinite) relational structures. In what follows, we present the characterization from~\cite{DBLP:conf/pods/ConsoleKP21} for databases, i.e., finite relational structures. Apart from the standard properties of $1$-criticality and closure under direct products, already used in Section~\ref{sec:ftgds-egds}, we are going to use a novel locality property. The latter locality property was introduced in~\cite{DBLP:conf/pods/ConsoleKP21} for arbitrary structures, but we are going to adapt it for databases.
We proceed to introduce this notion of locality and show that every collection $C$ of databases that is finitely axiomatizable by tgds and egds enjoys this property. We then present the characterization of arbitrary tgds and egds and discuss how it allows us to derive the characterization for full tgds and egds from~\cite{DBLP:journals/acta/MakowskyV86}, namely Theorem~\ref{the:tgd+egd-MV86}.

\subsection{Locality}\label{sec:locality}

The locality property of interest relies on the notion of local embedding of a collection of databases in a database.
Roughly speaking, a collection $C$ of databases over a schema $\ins{S}$ is locally embeddable in an $\ins{S}$-database $D$ if, for every subdatabase $E$ of $D$ with a bounded number of active domain elements (i.e., domain elements that occur in $\facts{E}$), we can find a database $D_E \in C$ such that every local neighbour of $E$ in $D_E$ (i.e., subdatabases of $D_E$ that contain $E$ and have a bounded number of additional active domain elements that do not occur in $\facts{E}$), can be embedded in $D$ while preserving $E$.
We call the collection $C$ of databases local if, for every $\ins{S}$-database $D$, $C$ is locally embeddable in $D$ implies that $D$ belongs to $C$.
We proceed to formalize the above high-level description.

\begin{figure}[t]
	\centering
	\includegraphics[scale=0.7]{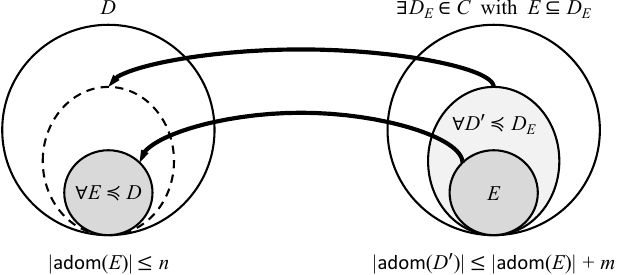}
	\caption{$C$ is $(n,m)$-locally embeddable in $D$.}
	\label{fig:embedding}
\end{figure}

Recall that the active domain of a database $D$, denoted $\aadom{D}$, is the set of elements of $\adom{D}$ that occur in at least one fact of $D$.
Consider an $\ins{S}$-database $D'$ and an $\ins{S}$-database $D'' \subseteq D'$, i.e., $\facts{D''} \subseteq \facts{D'}$. The {\em $m$-neighbourhood of $D''$ in $D'$} is defined as the set of $\ins{S}$-databases
\[
\{E \mid \aadom{D''} \subseteq \aadom{E},\ E \preceq D' \text{ and } |\aadom{E}| \leq |\aadom{D''}| + m\},
\]
that is, all the subdatabases of $D'$ such that their facts contain constants from $\aadom{D''}$ and at most $m$ additional elements not occurring in $\aadom{D''}$.
It is easy to verify that the $m$-neighbourhood of $D''$ in $D'$ is actually the set of $\ins{S}$-databases
\[
\{E \mid D'' \subseteq E,\ E \preceq D' \text{ and } |\aadom{E}| \leq |\aadom{D''}| + m\},
\] 
i.e., the set all subdatabases of $D'$ that contain $D''$ and their facts mention at most $m$ additional elements not occurring in the facts of $D''$. Indeed, if $D'' \subseteq D'$ and $E \preceq D'$, then we have that $\aadom{D''} \subseteq \aadom{E}$ if and only if $D'' \subseteq E$.

Consider a collection $C$ of databases over $\ins{S}$ and an $\ins{S}$-database $D$. For integers $n,m \geq 0$, we say that {\em $C$ is $(n,m)$-locally embeddable in $D$} if, for every $E \preceq D$ with $|\aadom{E}| \leq n$, there is $D_E \in C$ such that $E \subseteq D_E$, and for every $D'$ in the $m$-neighbourhood of $E$ in $D_E$, there is a function $h : \aadom{D'} \ra \aadom{D}$
such that $h$ is the identity on $\aadom{E}$ and  $h(\facts{D'}) \subseteq \facts{D}$.
The notion of $(n,m)$-local embeddability is illustrated in Figure~\ref{fig:embedding}; the circles represent the set of facts of the databases.
Clearly, for every $\ins{S}$-database $E$, we have that $\aadom{E} \subseteq \adom{E}$, and this containment may be a proper one. Observe, however, that the notion of a collection $C$ of databases being $(n,m)$-locally embeddable in $D$ only depends on $\aadom{E}$ of $E$ and the set of facts of $E$, and not on $\adom{E}$ of $E$. In turn, this implies that, when showing that $C$ is $(n,m)$-locally embeddable in $D$, it suffices to focus our attention on subdatabases $E$ of $D$ such that $\aadom{E} = \adom{E}$. We state this observation as a separate lemma, which will be used later in the paper.

\begin{lemma} \label{lem:adom} 
	Consider a collection $C$ of databases, integers $n,m \geq 0$, and a database $D$. The following statement are equivalent:
	\begin{enumerate}
		\item $C$ is $(n,m)$-locally embeddable in $D$. 
		\item For every $E \preceq D$ with $\aadom{E} = \adom{E}$ and  $|\aadom{E}| \leq n$, there is $D_E \in C$ such that $E \subseteq D_E$, and for every $D'$ in the $m$-neighbourhood of $E$ in $D_E$, there is a function $h : \aadom{D'} \ra \aadom{D}$ such that $h$ is the identity on $\aadom{E}$ and  $h(\facts{D'}) \subseteq \facts{D}$.
	\end{enumerate}
\end{lemma}

We are now ready to give the definition of the central notion of locality.

\begin{definition}[Locality]\label{def:locality}
	A collection $C$ of databases over $\ins{S}$ is {\em $(n,m)$-local}, for $n,m \geq 0$, if, for every $\ins{S}$-database $D$, the following holds: $C$ is $(n,m)$-locally embeddable in $D$ implies $D \in C$. We say that $C$ is {\em local} if it is $(n,m)$-local for some $n,m \geq 0$.
\end{definition}

We can now show the following technical lemma that essentially states that every collection of databases that is finitely axiomatizable by tgds with at most $n$ universally and $m$ existentially quantified variables, and egds with at most $n$ universally quantified variables, is $(n,m)$-local. 
A tgd is called {\em $(n,m)$-tgd}, for $n \geq 0$ and $m >0$, or $n > 0$ and $m \geq 0$, if it mentions at most $n$ universally and $m$ existentially quantified variables.
Moreover, an egd is called {\em $n$-egd}, for $n >0$, if it mentions at most $n$ universally quantified variables.
For notational convenience, we also define the corner cases of $(0,0)$-tgd and $0$-egd as the truth value $\mathsf{true}$, i.e., as a tautology.

\begin{figure}[t]
	\centering
	\includegraphics[scale=0.7]{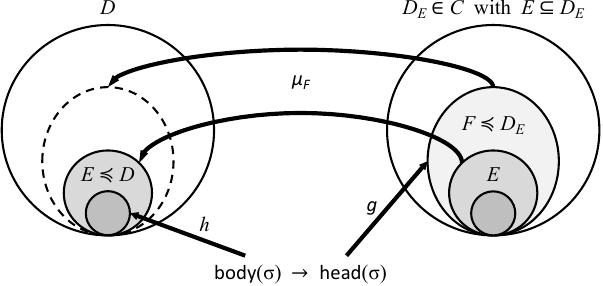}
	\caption{The function $\lambda = \mu_L \circ g$ in the proof of Lemma~\ref{lem:locality}.}
	\label{fig:proof}
\end{figure}

\begin{lemma}\label{lem:locality}
	For integers $n,m \geq 0$, every collection $C$ of databases that is finitely axiomatizable by $(n,m)$-tgds and $n$-egds is $(n,m)$-local.
\end{lemma}

\begin{proof}
	Let $C$ be a collection of databases over a schema $\ins{S}$ that is finitely axiomatizable by $(n,m)$-tgds and $n$-egds. By definition, there is a finite set $\dep_T$ of $(n,m)$-tgds and a finite set $\dep_E$ of $n$-egds such that, for every $\ins{S}$-database $D$, we have that $D \in C$ iff $D \models \dep_T$ and $D \models \dep_E$.
	Consider an $\ins{S}$-database $D$ and assume that $C$ is $(n,m)$-locally embeddable in $D$. We proceed to show that $D \in C$, i.e., $D \models \dep_T$ and $D \models \dep_E$.
	
	We first show that $D \models \dep_T$. Consider a tgd $\sigma \in \dep_T$, other than the $(0,0)$-tgd that is trivially satisfied by $D$, of the form $\phi(\bar x, \bar y) \ra \exists \bar z \, \psi(\bar x, \bar z)$, and assume that there exists a function $h : \bar x \cup \bar y \ra \adom{D}$ such that $h(\phi(\bar x, \bar y)) \subseteq \facts{D}$. We show that there exists an extension $\lambda$ of $h$ such that $\lambda(\psi(\bar x, \bar z)) \subseteq \facts{D}$; the existence of $\lambda$ is illustrated in Figure~\ref{fig:proof}.
	Let $E$ be the database $(\adom{E},R_{1}^{E},\ldots,R_{\ell}^{E})$, where $\adom{E}$ is the set of constants occurring in $h(\phi(\bar x,\bar y))$, and, for each $i \in [\ell]$, $R_{i}^{E} = {R_{i}^{D}}_{|E}$. It is clear that $E \preceq D$ with $|\aadom{E}| \leq n$ since $\phi(\bar x,\bar y)$ mentions at most $n$ variables. Since, by hypothesis, $C$ is $(n,m)$-locally embeddable in $D$, we get that there exists $D_E \in C$ such that $E \subseteq D_E$, and, for every $D'$ in the $m$-neighbourhood of $E$ in $D_E$, there is a function $\mu_{D'} : \aadom{D'} \ra \aadom{D}$, which is the identity on $\aadom{E}$, such that $\mu_{D'}(\facts{D'}) \subseteq \facts{D}$. It is clear that $h(\phi(\bar x,\bar y)) \subseteq \facts{D_E}$. Since $D_E \in C$, or, equivalently, $D_E \models \dep$, there exists an extension $g$ of $h$ such that $g(\psi(\bar x,\bar z)) \subseteq \facts{D_E}$.
	Let $F = (\adom{F},R_{1}^{F},\ldots,R_{\ell}^{F})$, where $\adom{F}$ are the constants occurring in $h(\phi(\bar x,\bar y)) \cup g(\psi(\bar x,\bar z))$, and, for each $i \in [\ell]$, $R_{i}^{F} = {R_{i}^{D_E}}_{|\adom{F}}$. It is clear that $F$ is in the $m$-neighbourhood of $E$ in $D_E$ since $\bar z$ has at most $m$ variables. Therefore, there is a function $\mu_F :\aadom{F} \ra \aadom{D}$, which is the identity on $\aadom{E}$, such that $\mu_F(\facts{F}) \subseteq \facts{D}$.
	Consider the function $\lambda = \mu_F \circ g$. Since $g$ is an extension of $h$, and $\mu_F$ is the identity on the elements occurring in $h(\phi(\bar x,\bar y))$, we get that $\lambda(v) = h(v)$ for each variable $v$ in $\phi(\bar x,\bar y)$, and thus, $\lambda$ is an extension of $h$. Moreover, since $g(\psi(\bar x,\bar z)) \subseteq \facts{F}$, we get that $\lambda(\psi(\bar x,\bar z)) \subseteq \facts{D}$.
	
	We now show that $D \models \dep_E$. Consider an egd $\theta \in \dep_E$, other than the $0$-egd that is trivially satisfied by $D$, of the form $\phi(\bar x) \ra x_i = x_j$, and assume there exists a function $h : \bar x \ra \adom{D}$ such that $h(\phi(\bar x)) \subseteq \facts{D}$. We show that $h(x_i) = h(x_j)$.
	Let $E = (\adom{E},R_{1}^{E},\ldots,R_{\ell}^{E})$, where $\adom{E}$ is the set of constants occurring in $h(\phi(\bar x))$, and, for each $i \in [\ell]$, $R_{i}^{E} = {R_{i}^{D}}_{|E}$. It is clear that $E \preceq D$ with $|\aadom{E}| \leq n$ since $\phi(\bar x)$ mentions at most $n$ variables. Since, by hypothesis, $C$ is $(n,m)$-locally embeddable in $D$, there exists $D_E \in C$ such that $E \subseteq D_E$. Observe that $E \subseteq D_E$ implies that $h(\phi(\bar x)) \subseteq E_E$. Since $D_E \models \dep_E$, we get that $h(x_i) = h(x_j)$.
\end{proof}

It is useful to observe that locality implies domain independence.

\begin{lemma}\label{lem:domain-independence}
	If a collection of databases is local, then it is also domain independent.
\end{lemma}

\begin{proof}
	Assume that $C$ is a collection of databases over $\ins{S}$ that is $(n,m)$-local for $n,m \geq 0$, and let $D$ be an $\ins{S}$-database in $C$.
	Consider an $\ins{S}$-database $D'$ such that $\facts{D} = \facts{D'}$. We need to show that $D' \in C$.
	This is done by showing that $C$ is $(n,m)$-locally embeddable in $D'$, which implies that $D' \in C$ since, by hypothesis, $C$ is $(n,m)$-local.
	Consider an $\ins{S}$-database $D'' \preceq D'$ with $|\aadom{D''}| \leq n$. Since $\facts{D} =\facts{D'}$, it is clear that $D'' \subseteq D$. Since $D \in C$, it suffices to show that, for every $E$ in the $m$-neighbourhood of $D''$ in $D$, there is a function $h_{E} : \aadom{E} \ra \aadom{D'}$, which is the identity on $\aadom{D''}$, such that $h_{E}(\facts{E}) \subseteq \facts{D'}$. Since $E \subseteq D$, which means that $\facts{E} \subseteq \facts{D}$, we get that $\facts{E} \subseteq \facts{D'}$. Therefore, $h_{E}$ is simply the identity on $\aadom{E}$.
\end{proof}

\subsection{The Characterization for TGDs and EGDs}

We are now ready to provide the characterization of interest, which was originally established for arbitrary relational structures in~\cite{DBLP:conf/pods/ConsoleKP21}, and here is presented for databases (i.e., finite structures).

\begin{theorem}\label{the:tgd-egd-characterization}
	Consider a collection $C$ of databases, and integers $n,m \geq 0$. The following statements are equivalent:
	\begin{enumerate}
		\item $C$ is finitely axiomatizable by $(n,m)$-tgds and $n$-egds.
		\item $C$ is $1$-critical, $(n,m)$-local, and closed under direct products.
	\end{enumerate}
\end{theorem}

Clearly, the above result implies a characterization for arbitrary tgds and egds:

\begin{corollary} \label{cor:tgd+egd}
	Consider a collection $C$ of databases. The following are equivalent:
	\begin{enumerate}
		\item $C$ is finitely axiomatizable by tgds and egds.
		\item $C$ is $1$-critical, local, and closed under direct products.
	\end{enumerate}
\end{corollary}

The direction $(1) \Rightarrow (2)$ of Theorem~\ref{the:tgd-egd-characterization} is easy and we leave it as an exercise to the reader. We proceed to discuss the non-trivial direction $(2) \Rightarrow (1)$. To this end, we need to introduce a couple of auxiliary notions, namely existential disjunctive dependencies and relative diagrams of databases.

\medskip

\noindent \textbf{Existential Disjunctive Dependencies.} Existential disjunctive dependencies generalize disjunctive dependencies, already used in the proof of Theorem~\ref{the:ftgd+egd-intersections}, with existential quantification in the right-hand side of the implication. 
Moreover, the hypothesis and the conclusion of the implication can be empty.
More precisely, an {\em existential disjunctive dependency} (edd) $\delta$ over a schema $\ins{S}$ is a constant-free sentence
\[
\forall \bar x \left(\phi(\bar x)\ \ra\ \bigvee_{i=1}^{k}
\psi_i(\bar x_i)\right),
\]
where $\bar x$ is a tuple of variables of $\ins{V}$, the expression $\phi(\bar x)$ is a (possibly empty) conjunction of atoms over $\ins{S}$, and, for each $i \in [k]$, $\bar x_i \subseteq \bar x$ and the expression $\psi(\bar x_i)$ is either an equality formula $y=z$ with $\bar x_i = \{y,z\}$, or a constant-free formula $\exists \bar y_i \, \chi_i(\bar x_i,\bar y_i)$ with $\bar y_i$ being a tuple of variables from $\ins{V} \setminus \bar x$ and $\chi_i(\bar x_i,\bar y_i)$ is a (possibly empty) conjunction of atoms over $\ins{S}$.
When the conclusion of $\delta$ is empty (i.e., there are no disjuncts), $\delta$ is essentially the sentence $\forall \bar x \left(\phi(\bar x) \ra \mathsf{false}\right)$. Now, when both the hypothesis and the conclusion of $\delta$ are empty, then $\delta$ is essentially the truth value $\mathsf{false}$, i.e., a contradiction.
Assuming that the conclusion of $\delta$ is non-empty, a database $D$ \emph{satisfies} $\delta$ if, whenever there exists a function $h : \bar x \ra \adom{D}$ such that $h(\phi(\bar x)) \subseteq \facts{D}$, then there is $i \in [k]$ such that, if $\psi_i(\bar y_i)$ is $y=z$, then $h(y) = h(z)$; otherwise, if $\psi_i(\bar y_i)$ is of the form $\exists \bar y_i \, \chi_i(\bar x_i,\bar y_i)$, then there is an extension $h'$ of $h$ such that $h'(\chi_i(\bar x_i,\bar y_i)) \subseteq \facts{D}$. 
In case the conclusion of $\delta$ is empty and $\delta$ is not a contradiction, $D$ \emph{satisfies} $\delta$ if there is no function $h : \bar x \ra \adom{D}$ such that $h(\phi(\bar x)) \subseteq \facts{D}$.
We write $D \models \delta$ for the fact that $D$ satisfies $\delta$. The database $D$ \emph{satisfies a set} $\dep$ of edds, written $D \models \dep$ ($D$ is a {\em model} of $\dep$), if $D \models \delta$ for each $\delta \in \dep$.

\medskip
\noindent \textbf{Relative Diagram of a Database.} Consider an integer $\ell \geq 0$, an $\ins{S}$-database $E$, and an $\ins{S}$-database $D$ such that $E \preceq D$.
We are going to define the notion of the \emph{$\ell$-diagram of $E$ relative to $D$}, which can be regarded as a refinement of the standard notion of diagram of a database, already used in the proof of Theorem~\ref{the:ftgd+egd-intersections}.
Let $A_{E,\ell}$ be the set of all atomic formulas of the form $R(\bar u)$ that can be formed using predicates from $\ins{S}$, constants from $\adom{E}$, and $\ell$ distinct variables $y_1,\ldots,y_\ell$ from $\ins{V}$, i.e., $R \in \ins{S}$ and $\bar u \in (\adom{E} \cup \{y_1,\ldots,y_\ell\})^{\ar{R}}$.  
Let $C_{E,\ell}$ be the set of all conjunctions of atomic formulas from $A_{E,\ell}$. Note that both $A_{E,\ell}$ and
$C_{E,\ell}$ are finite sets because $\adom{E}$ is finite.
Given a formula $\gamma(y_1,\ldots,y_\ell) \in C_{E,\ell}$, we can naturally talk about the satisfaction of the sentence $\exists y_1 \cdots \exists y_\ell\, \gamma( y_1,\ldots,y_\ell)$ by a database $D'$, in which case we simply write $D' \models \exists y_1 \cdots \exists y_\ell \, \gamma(y_1,\ldots,y_\ell)$.
The {\em $\ell$-diagram of $E$ relative to $D$}, denoted $\Delta_{E,\ell}^{D}$, is the first-order sentence
{\small 
\[
\bigwedge_{\alpha \in \facts{E}} \alpha\ \wedge\ \bigwedge_{\substack{c,d \in \adom{E}, \\ c \neq d}} \neg (c=d)\ \wedge\ \bigwedge_{\substack{\gamma(y_1,\ldots,y_\ell) \in C_{E,\ell}, \\ D \not\models \exists y_1 \cdots \exists y_\ell \, \gamma(y_1,\ldots,y_\ell)}} \neg \left(\exists y_1 \cdots \exists y_\ell \, \gamma(y_1,\ldots,y_\ell)\right).
\]
}
In fact, we are interested in the first-order formula $\Phi_{E,\ell}^{D}(\bar x)$ obtained from the first-order sentence $\Delta_{E,\ell}^{D}$ by replacing each constant element $c \in \adom{E}$ with a new variable $x_c \in \ins{V} \setminus \{y_1,\ldots, y_\ell\}$.
As was the case with the formulas of $C_{E,\ell}$, we can naturally talk about the satisfaction of $\exists \bar x \, \Phi_{E,\ell}^{D}(\bar x)$ by a database $D'$, in which case we simply write $D' \models \exists \bar x \, \Phi_{E,\ell}^{D}(\bar x)$. It is straightforward to verify that:

\begin{lemma}\label{lem:relative-diagram}
	For every integer $\ell \geq 0$, $\ins{S}$-database $E$, and $\ins{S}$-database $D$ with $E \preceq D$, it holds that $D \models \exists \bar x \, \Phi_{E,\ell}^{D}(\bar x)$.
\end{lemma}

Let $\class{E}_{n,m}$, for some integers $n,m \geq 0$, be the set of all edds over $\ins{S}$ of the form $\forall \bar x \left(\phi(\bar x) \ra \bigvee_{i=1}^{k} \psi_i(\bar x_i)\right)$ such that $\bar x$ consists of at most $n$ distinct variables, and, for each $i \in [k]$, the formula $\psi_i(\bar x_i)$ mentions at most $n+m$ distinct variables. The latter means that, if $\psi_i(\bar x_i)$ is a formula of the form $\exists \bar y_i \, \chi_i(\bar x_i,\bar y_i)$ (i.e., it is not an equality expression), then $\bar y_i$ consists of at most $m$ distinct variables. Note that $\class{E}_{n,m}$ is a finite set (up to logical equivalence) since $\ins{S}$ is finite and the number of variables in each element of $\class{E}_{n,m}$ is finite. We can show the following auxiliary lemma.

\begin{lemma}
	\label{lem:tgd-egd-dc-characterization-lemma-1-B}
	Let $n,m \geq 0$ and assume that  $D,E$ are two $\ins{S}$-databases such that $\aadom{E} = \adom{E}$, $|\aadom{E}|\leq n$, and $E \preceq D$. It holds that there exists a sentence $\delta \in \class{E}_{n,m}$ such that $\delta \equiv \neg \exists \bar x \, \Phi_{E,m}^{D}(\bar x)$.
\end{lemma}

\begin{proof}
	Recall that the formula $\Phi_{E,m}^{D}(\bar x)$ is obtained from the $m$-diagram of $E$ relative to $D$ by renaming each $c \in \adom{E}$ to a new variable $x_c$; let $\rho$ be the renaming function, i.e., $\rho(c) = x_c$ for each $c \in \adom{E}$. Thus, $\Phi_{E,m}^{D}(\bar x)$ is of the form
	\[
	\underbrace{\bigwedge_{\alpha \in \facts{E}} \rho(\alpha)}_{\Psi_1}\ \wedge \underbrace{\bigwedge_{\substack{c,d \in \adom{E}, \\ c \neq d}} \neg (\rho(c)=\rho(d))\ \wedge \bigwedge_{\substack{\gamma(\bar y) \in C_{E,m}, \\ D \not\models \exists \bar y \, \gamma(\bar y)}} \neg \exists \bar y \, \rho(\gamma(\bar y))}_{\Psi_2}.
	\]
	Note that $\Psi_1$ and $\Psi_2$ might be empty (i.e., they have no conjuncts). In particular, $\Psi_1$ is empty if $E$ is empty, whereas $\Psi_2$ is empty if $D$ is $1$-critical, and thus, $D \models \exists \bar y \, \gamma(\bar y)$ for each $\gamma(\bar y) \in C_{E,m}$.
	It is clear that $\neg \exists \bar x \,\Phi_{E,m}^{D}(\bar x)$ is equivalent to the sentence
	\[
	\delta\ =\ \forall \bar x \left(\phi(\bar x)\ \ra\ \psi(\bar x)\right),
	\] 
	where
	\[
	\phi(\bar x)\ =\ \bigwedge_{\alpha \in \facts{E}} \rho(\alpha)
	\]
	and the shape of $\psi(\bar x)$ depends on whether $\Psi_1$ and $\Psi_2$ are empty or not. 
	If both are empty, then $\delta$ is the truth constant $\mathsf{false}$, i.e., a contradiction, and thus, $\delta \in \class{E}_{n,m}$.
	If $\Psi_1$ is non-empty and $\Psi_2$ is empty, then, for an $\ins{S}$-database $D'$, it holds that $D' \models \neg \exists \bar x \,\Phi_{E,m}^{D}(\bar x)$ if there is no function $h : \bar x \ra \adom{D'}$ such that $h(\phi(\bar x)) \subseteq \facts{D'}$. Hence, in this case, $\delta$ is the edd $\forall \bar x \, \left(\phi(\bar x) \ra \mathsf{false}\right)$. Since, by hypothesis, $|\aadom{E}| \leq n$, we get that $\phi(\bar x)$ mentions at most $n$ variables, and thus, $\delta \in \class{E}_{n,m}$.
	Now, if $\Psi_1$ is either empty or non-empty, and $\Psi_2$ is non-empty, then we have that
	\[
	\psi(\bar x)\ =\ \bigvee_{\substack{c,d \in \adom{E}, \\ c \neq d}} \rho(c) = \rho(d) \vee \bigvee_{\substack{\gamma(\bar y) \in C_{E,m}, \\ D \not\models \exists \bar y \, \gamma(\bar y)}} \exists \bar y \, \rho(\gamma(\bar y)),
	\]
	and hence, $\delta$ is an edd.	
	It remains to show that $\delta \in \mathsf{E}_{n,m}$. To this end, we need to show the following two statements: (i)  each variable in $\psi(\bar x)$ is either existentially quantified or appears in $\phi(\bar x)$, and (ii) $\delta$ mentions at most $n$ universally quantified variables and at most $m$ existentially quantified variables. Statement (i) is true since, by hypothesis, $\aadom{E} = \adom{E}$. Concerning statement (ii), $\delta$ has at most $n$ universally quantified variables since, by hypothesis, $|\aadom{E}| \leq n$, and $\delta$ has $m$ existentially quantified variables because so does the formula $\neg \Phi_{E,m}^D$ by the construction of $\Phi_{E,m}^D$.
\end{proof}

We can now discuss the non-trivial direction $(2) \Rightarrow (1)$ of Theorem~\ref{the:tgd-egd-characterization}. Assume that $C$ is a collection of databases over $\ins{S}$. The proof proceeds in two main steps:

\begin{enumerate}
	\item We construct a finite set $\dep^{\vee}$ of edds over $\ins{S}$, which mention at most $n$ universally and at most $m$ existentially quantified variables, and show that $\dep^{\vee}$ has the following property:  a database $D$ is in $C$ if and only if  $D$ satisfies  $\dep^{\vee}$. To this end, we exploit the fact that $C$ is $(n,m)$-local.
	
	\item We then show that there is a finite set $\dep^{\exists,=}$ of $(n,m)$-tgds and $n$-egds over $\ins{S}$ such that a database $D$ is in $C$ if and only if $D$ satisfies $\dep^{\exists,=}$;
 in fact, $\dep^{\exists,=}$ is the set of the tgds and egds occurring in $\dep^{\vee}$. To this end, we exploit the fact that $C$ is $1$-critical and closed under direct products.
\end{enumerate}
In what follows, we give the details for each of the above steps.

\bigskip
\noindent
\textbf{\underline{Step 1: The finite set $\dep^{\vee}$ of edds that axiomatizes $C$}}

\smallskip
\noindent Recall that $\class{E}_{n,m}$ is the set of all edds over $\ins{S}$ of the form $\forall \bar x \left(\phi(\bar x) \ra \bigvee_{i=1}^{k} \psi_i(\bar x_i)\right)$ such that $\bar x$ consists of at most $n$ distinct variables, and, for each $i \in [k]$, the formula $\psi_i(\bar x_i)$ mentions at most $n+m$ distinct variables. The latter means that, if $\psi_i(\bar x_i)$ is a formula of the form $\exists \bar y_i \, \chi_i(\bar x_i,\bar y_i)$ (i.e., it is not an equality expression), then $\bar y_i$ consists of at most $m$ distinct variables. Recall also that $\class{E}_{n,m}$ is finite (up to logical equivalence).
Define $\dep^{\vee}$ to be the set of all edds from $\class{E}_{n,m}$ that are satisfied by every database of $C$, that is,
\[
\dep^{\vee}\ =\ \left\{\delta \in \class{E}_{n,m} \mid \text{ for each } D \in C, \text{ it holds that } D \models \delta\right\}.
\]
It is clear that $\dep^{\vee}$ is finite (up to logical equivalence) since $\dep^{\vee} \subseteq \class{E}_{n,m}$. We will show that $C$ is precisely the set of databases that satisfy $\dep^{\vee}$.

\begin{lemma}\label{lem:tgd-egd-dc-characterization-lemma-1}
	For every $\ins{S}$-database $D$, we have that  $D \in C$ if and only if $D \models \dep^{\vee}$.
\end{lemma}

\begin{proof}
	The $(\Rightarrow)$ direction follows from the definition of $\dep^{\vee}$.
	We proceed with the $(\Leftarrow)$ direction. We have to show that if $D \models \dep^{\vee}$, then $D \in  C$. We will show that $C$ is $(n,m)$-locally embeddable in $D$, which, since $C$ is $(n,m)$-local, will imply that $D \in C$.
	By Lemma \ref{lem:adom}, to show that $C$ is $(n,m)$-locally embeddable in $D$, we have to show that, for every $E \preceq D$ with $\aadom{E} = \adom{E}$ and  $|\aadom{E}| \leq n$, there is a database $D_E \in C$ such that $E \subseteq D_E$, and for every $D'$ in the $m$-neighbourhood of $E$ in $D_E$, there is a function $h : \aadom{D'} \ra \aadom{D}$ such that $h$ is the identity on $\aadom{E}$ and $h(\facts{D'}) \subseteq \facts{D}$.

	Let $E$ be a database such that $E \preceq D$, $\aadom{E} = \adom{E}$, and $|\aadom{E}| \leq n$. By Lemma~\ref{lem:tgd-egd-dc-characterization-lemma-1-B}, the sentence $\neg \exists \bar x \, \Phi^D_{E,m}(\bar x)$ is logically equivalent to a sentence $\delta \in \class{E}_{n,m}$. We claim that $\delta \not \in \dep^{\vee}$. Indeed, otherwise, we would have that $D \models \delta$, which is a contradiction, since, by Lemma \ref{lem:relative-diagram}, $D \models \exists \bar x \, \Phi^D_{E,m}(\bar x)$. Since $\delta \not \in \dep^{\vee}$, there must exist a database $D_E \in C$ such that $D_E \not \models \delta$, which means that $D_E \models \exists \bar x \, \Phi^D_{E,m}(\bar x)$.
	Thus, there exists a database $E' \subseteq D_E$ such that $E \simeq E'$, i.e., $E$ and $E'$ are isomorphic databases. Therefore, and without loss of generality, we can assume that $E \subseteq D_E$.
	To show that $C$ is $(n,m)$-locally embeddable in $D$, it remains to show that for every database $D'$ in the $m$-neighbourhood of $E$ in $D_E$, there exists a function $h : \aadom{D'} \ra \aadom{D}$ such that $h$ is the identity on $\aadom{E}$ and $h(\facts{D'}) \subseteq \facts{D}$.
	Towards a contradiction, assume that there exists $D'$ in the $m$-neighbourhood of $E$ in $D_E$ for which there is no function $h : \aadom{D'} \ra \aadom{D}$ that is the identity on $\aadom{E}$ with $h(\facts{D'}) \subseteq \facts{D}$.
	Let $F$ be the $\ins{S}$-database defined as the difference between $D'$ and $E$, i.e., $F$ is such that $\facts{F} = \facts{D'} \setminus \facts{E}$, while $\adom{F}$ consists of all the constants occurring in $\facts{D'} \setminus \facts{E}$, and hence, $\adom{F} = \aadom{F}$.
	Clearly, there is no function $h : \aadom{F} \ra \aadom{D}$ that is the identity on $\aadom{E}$ and $h(\facts{F} \subseteq D$.
	Observe that $|\aadom{F} \setminus \aadom{E}| \leq m$; we assume that $\aadom{F} \setminus \aadom{E} = \{d_1,\ldots,d_{m'}\}$ for $m' \leq m$.
	Let $\gamma(\bar y)$ be the formula obtained from $\bigwedge_{\alpha \in \facts{F}} \alpha$ after renaming each constant $d_i$ to the variable $y_i$; thus, $\bar y = y_1,\ldots,y_{m'}$.
	Since there is no function $h : \aadom{F} \ra \aadom{D}$ that is the identity on $\aadom{E}$ and $h(\facts{F} \subseteq D$, we can conclude that $D \not\models \exists \bar y \, \gamma(\bar y)$. Observe now that, by construction, $\neg \exists \bar y \, \gamma(\bar y)$ is a conjunct of $\Delta_{E,m}^{D}$, which in turn implies that the formula $\neg \exists \bar z \, \gamma(\bar z)$ obtained from $\neg \exists \bar y \, \gamma(\bar y)$ after renaming each constant $c \in \adom{E} = \aadom{E}$ to the variable $x_c$ is a conjunct of $\Phi_{E,m}^{D}(\bar x)$. Since $F \subseteq D_E$, we conclude that $D_E \models \exists \bar z \, \gamma(\bar z)$, which implies that $D_E \not\models \exists \bar x \, \Phi_{E,m}^{D}(\bar x)$. But this contradicts the fact that $D_E \models \exists \bar x \, \Phi_{E,m}^{D}(\bar x)$. This completes the proof that $C$ is $(n,m)$-embeddable in $D$.
\end{proof}

\bigskip
\noindent
\textbf{\underline{Step 2: The finite set $\dep^{\exists,=}$ of tgds and egds that axiomatizes $C$}}

\smallskip
\noindent 
The next lemma, which is an elimination-of-disjunctions result,  provides a stepping stone towards proving that  the class $C$ is finitely axiomatizable by tgds and egds.

\begin{lemma}\label{lem:tgd-egd-dc-characterization-lemma-2}
Assume that 
	$
	\forall \bar x (\phi(\bar x) \ra \bigvee_{i=1}^{k} \psi_i(\bar x_i)),
	$
	where, for each $i \in [k]$, $\psi_i(\bar x_i)$ is either an equality formula or a non-empty conjunction of atoms, is an edd that belongs to $\dep^\vee$.
	There is $j \in [k]$ such that $
		\forall \bar x (\phi(\bar x)\rightarrow \psi_j(\bar x_j))$
  also belongs to $\dep^\vee$.
\end{lemma}

\begin{proof} 
	Note that for every $j\in [k]$,
the sentence $\forall \bar x (\phi(\bar x)\rightarrow \psi_j(\bar x_j))$ is either an egd or a tgd, hence it is an edd. Towards a contradiction, assume that none of these sentences belongs to $\dep^\vee$. Since $\dep^\vee$ consists of all edds in $\class{E}_{n,m}$ that are satisfied by every database in $C$, 
for every $j \in [k]$,
		there is an $\ins{S}$-database $D_j \in C$ such that $D_j \not\models \forall \bar x \, (\phi(\bar x)\rightarrow \psi_j(\bar x_j)) $, or, equivalently, $D_j \models \exists \bar x \, (\phi(\bar x) \wedge \neg \psi_j(\bar x_j))$.
		Let
		\[
		E\ =\ D_{1} \otimes \cdots \otimes D_{k}. 
		\]
		Since, by hypothesis, $C$ is closed under direct products, we have that $E \in C$. We will show that $E \not\models \delta$, which leads to a contradiction since $\delta \in \Sigma^{\vee}$ and $E \in C$, which means that $E \models \delta$.

		Since $D_j \models \exists \bar x \, (\phi(\bar x) \wedge \neg \psi_j(\bar x_j))$ for each $j \in [k]$, there is a function $h_j : \bar x \ra \adom{D_j}$ such that $h_j(\phi(\bar x)) \subseteq \facts{D_j}$ and $D_j \models \neg \psi_j(h_j(\bar x_j))$. Let $h : \bar x \ra \adom{E}$ be the function such that, for each variable $x \in \bar x$, we have $h(x)\ =\ (h_1(x),\ldots,h_k(x))$. By the definition of direct products, we have that $h(\phi(\bar x)) \subseteq \facts{E}$. It remains to show that, for each $j \in [k]$, if $\psi_j(\bar x_j)$ is an equality formula $y=z$, then $h(y) \neq h(z)$, while if $\psi_j(\bar x_j)$ is not an equality formula, then there is no extension $h'$ of $h$ such that $h'(\psi_j(\bar x_j)) \subseteq \facts{E}$.
		We proceed by case analysis on the form of $\psi_j(\bar x_j)$.
		
		\medskip
		
		\noindent \textbf{Case 1:} Assume first that $\psi_j(\bar x_j)$ is the equality expression $y=z$.
		By the properties of the function $h_j$, we have that $h_j(y)\not = h_j(z)$. Hence, by the definition of $h$, we conclude that $h(y) \not =  h(z)$.
		
		\medskip
		\noindent \textbf{Case 2:}
		Assume now that $\psi_j(\bar x_j)$ is a formula of the form $\exists \bar y_j \, \chi_j(\bar x_j,\bar y_j)$, and assume, by contradiction, that there exists an extension $h'$ of $h$ such that $h'(\chi_j(\bar x_j,\bar y_j)) \subseteq \facts{E}$. For a constant $c \in \adom{E}$, we write $c[i]$ for the $i$-th component of $c$, which is actually a constant from $\adom{D_i}$.
		We define the function $h'_j : \bar x \cup \bar y_j \ra \adom{D_j}$ such that $h'_j(v) = h'(v)[j]$, for each variable $v \in \bar x \cup \bar y_j$. Since $h'$ is an extension of $h$ and $h_j(v) = h(v)[j]$, for each variable $v \in \bar x$, we conclude that $h'_j$ is an extension of $h_j$.
		By hypothesis, $h'(\alpha) \in \facts{E}$ for each conjunct $\alpha$ of $\chi_j(\bar x_j,\bar y_j)$. Assume that $h'(\alpha) = R(\bar c_i,\ldots,\bar c_m)$. By the definition of the direct product, we conclude that $R(\bar c_1[j],\ldots,c_m[j]) \in \facts{D_j}$. Consequently, $h'_j(\psi_j(\bar x_j)) \subseteq \facts{D_j}$, which in turn implies that $D_j \models \psi_j(h_j(\bar x_j)$, which contradicts the fact that $D_j \not \models \psi_j(h_j(\bar x_j)$. This completes the proof of Lemma~\ref{lem:tgd-egd-dc-characterization-lemma-2}.  
	\end{proof}

We are now ready to complete the proof of Theorem \ref{the:tgd-egd-characterization}.
Define $\dep^{\exists,=}$ to be the set of all tgds and egds occurring in $\dep^{\vee}$, that is,
\[
\dep^{\exists,=}\ =\ \left\{\delta \in \dep^{\vee} \mid \delta \text{ is a tgd or an egd} \right\}.
\]
Note that $\dep^{\exists,=}$ is a finite set because it is a subset of the finite set $\dep^\vee$.
We proceed to show that, for each $\ins{S}$-database $D$, it holds that $D\in C$ if and only if $D\models \dep^{\exists,=}$.



\medskip

($\Rightarrow$) If $D\in C$, then, by definition, $D\models \dep^\vee$; hence, $D \models \dep^{\exists,=}$ since $\dep^{\exists,=} \subseteq \dep^\vee$. 

($\Leftarrow$) Assume now that $D$ is an $\ins{S}$-database such that $D\models \dep^{\exists,=}$. We need to show that $D\in C$. By Lemma 
\ref{lem:tgd-egd-dc-characterization-lemma-1}, it suffices to show that $D\models \dep^\vee$. 
Let $\delta$ be a sentence in $\dep^{\vee}$.  Since $C$ is $1$-critical, we get that $\delta$ is not an implication with the truth constant $\mathsf{false}$ being its conclusion. Therefore, $\delta$ is an edd of the form	
	$
	\forall \bar x (\phi(\bar x) \ra \bigvee_{i=1}^{k} \psi_i(\bar x_i)),
	$
	where each expression $\psi_i(\bar x_i)$ is either an equality  $y=z$ with $y$ and $z$ among the variables in $\bar x$ or a constant-free formula $\exists \bar y_i \, \chi_i(\bar x_i, \bar y_i)$ with $\chi_i(\bar x_i, \bar y_i)$  non-empty, the variables in $\bar x_i$  among those in $\bar x$, and the variables in $\bar y_i$  different from those  in $\bar x$.  By Lemma \ref{lem:tgd-egd-dc-characterization-lemma-2}, 
 there is $j \in [k]$ such that the sentence $
		\forall \bar x (\phi(\bar x)\rightarrow \psi_j(\bar x_j))$ belongs to $\dep^\vee$. Since this sentence is either an egd or a tgd, we have that it belongs to $\dep^{=,\exists}$, and hence, $D\models \forall \bar x (\phi(\bar x)\rightarrow \psi_j(\bar x_j))$. It is obvious that the formula
  $\forall \bar x (\phi(\bar x)\rightarrow \psi_j(\bar x_j))$ logically implies the edd $\forall \bar x (\phi(\bar x) \ra \bigvee_{i=1}^{k} \psi_i(\bar x_i))$, hence $D\models \forall \bar x (\phi(\bar x) \ra \bigvee_{i=1}^{k} \psi_i(\bar x_i))$. This completes the proof of the above claim.

\medskip

Therefore, we have established that the class $C$ is axiomatizable by the set $\dep^{=, \exists}$, which is a finite set of tgds and egds. This completes the proof of Theorem \ref{the:tgd-egd-characterization} .

\subsection{From Theorem~\ref{the:tgd-egd-characterization} to Theorem~\ref{the:tgd+egd-MV86}}

We conclude this section by discussing how we can derive the characterization for full tgds and egds from~\cite{DBLP:journals/acta/MakowskyV86}, that is, Theorem~\ref{the:tgd+egd-MV86} in Section~\ref{sec:ftgds-egds}, by exploiting the characterization for $(n,m)$-tgds and $n$-egds presented above, that is, Theorem~\ref{the:tgd-egd-characterization}.
We first present the following consequence of Theorem~\ref{the:tgd-egd-characterization}, which forms a relevant characterization for finite sets of full tgds and egds in its own right.

\begin{corollary} \label{cor:ftgd+egd}
	Consider a collection $C$ of databases. The following are equivalent:
	\begin{enumerate}
		\item $C$ is finitely axiomatizable by full tgds and egds.
		\item $C$ is $1$-critical, $(n,0)$-local for some integer $n \geq 0$, and closed under direct products.
	\end{enumerate}
\end{corollary}

Observe that for obtaining Theorem~\ref{the:tgd+egd-MV86} it suffices to replace in Corollary~\ref{cor:ftgd+egd} the property of $(n,0)$-locality for some integer $n \geq 0$ with the properties of domain independence, modularity, and closure under subdatabases. We proceed to show that this is indeed the case. But let us first state a simple fact that provides a useful characterization of the notion of $(n,0)$-local embeddability.

Recall that, for a schema $\ins{S}$ and an integer $m \geq 0$, the $m$-neighbourhood of an $\ins{S}$-database $D''$ in an $\ins{S}$-database $D'$ with $D'' \subseteq D'$ is the set of $\ins{S}$-databases
\[
\{E \mid D'' \subseteq E,\ E \preceq D' \text{ and } |\aadom{E}| \leq |\aadom{D''}| + m\}.
\] 
Note that if $D'' \subseteq E$, we have that $\aadom{D''} \subseteq \aadom{E}$. Furthermore, if $E$ is in the $0$-neighbourhood of $D''$ in $D'$, then $|\aadom{E}| \leq |\aadom{D''}|$, and thus, $\aadom{D''} = \aadom{E}$. Since both $D''$ and $E$ are contained in $D'$, it follows that $\facts{E} = \facts{D'\restriction \aadom{D''}}=\facts{D''}$, where $D'\restriction \aadom{D''}$ is the restriction of $D'$ over $\aadom{D''}$, that is, the database $\{R(\bar c) \in D' \mid \bar c \in \aadom{D''}^{\ar{R}}\}$. It follows that the $0$-neighborhood of $D''$ in $D'$ is the set of $\ins{S}$-databases
\[
\{E \mid E \preceq D' \text{ and } \facts{E} = \facts{D'\restriction \aadom{D''}}=\facts{D''}\}.
\] 
Consequently, we have that the following holds. 

\begin{lemma}\label{fact:embeddable}
	Consider a collection $C$ of databases over a schema $\ins{S}$. The following are equivalent:
	\begin{enumerate} 
		\item An $\ins{S}$-database $D$ is $(n,0)$-locally embeddable in $C$.
		\item For every $E \preceq D$ with $\aadom{E} = \adom{E}$ and $|\aadom{E}| \leq n$, there is an $\ins{S}$-database $D_E \in C$ such that $\facts{D_E\restriction \aadom{E}}=\facts{E}$.
	\end{enumerate}
\end{lemma}

By exploiting the above lemma, we can show that $(n,0)$-locality can be replaced by the properties of domain independence, modularity, and  closure under subdatabases.

\begin{proposition} \label{prop:(n,0)-local-char} 
	Consider a collection $C$ of databases over a schema $\ins{S}$ and an integer $n \geq 0$. The following statements are equivalent:
	\begin{enumerate}
		\item $C$ is $(n,0)$-local.
		\item $C$ is domain independent, $n$-modular, and  closed under subdatabases.
	\end{enumerate}
\end{proposition}

\begin{proof} 
	We first show the direction $(1) \Rightarrow (2)$. Since, by hypothesis, $C$ is $(n,0)$-local, Lemma~\ref{lem:domain-independence} implies that $C$ is domain independent. 
	To show that $C$ is $n$-modular, let $D$ be an $\ins{S}$-database such that every $D' \preceq D$ with $\adom{D'} \leq n$ belongs to $C$. We proceed to show that $D \in C$ by showing that $C$ is $(n,0)$-locally embeddable in $D$. This follows from Lemma~\ref{fact:embeddable} by taking $D_E = D'$.  
	Finally, to show that $C$ is closed under subdatabases, consider an arbitrary $\ins{S}$-database $D \in C$ and let $D'$ be a subdatabase of $D$. We proceed to show that $D' \in C$ by showing that $C$ is $(n,0)$-locally embeddable in $D'$. The latter again follows from Lemma~\ref{fact:embeddable} by taking $D_E = D$.
	
	We now show $(2) \Rightarrow (1)$. Assume that $C$ is domain independent, $n$-modular, and closed under subdatabases. We need to show that if $C$ is $(n,0)$-locally embeddable in an $\ins{S}$-database $D$, then $D \in C$. By Lemma~\ref{fact:embeddable}, for every $E \preceq D$ with $\aadom{E} = \adom{E}$ and $|\aadom{E}|\leq n$, there is an $\ins{S}$-database $D_E \in C$ such that $\facts{D_E\restriction \aadom{E}}=\facts{E}$. Since $C$ is closed under subdatabases, we have that $D_E\restriction \aadom{E}$ belongs to $C$. Since $C$ is domain independent and $\facts{D_E\restriction \aadom{E}}=\facts{E}$, we get that $E \in C$. Thus, every subdatabase $E$ of $D$ with at most $n$ elements in its active domain is in $C$. Since $C$ is $n$-modular, it follows that $D$ is also in $C$.
\end{proof}

Theorem~\ref{the:tgd+egd-MV86} readily follows from Corollary~\ref{cor:ftgd+egd} and Proposition~\ref{prop:(n,0)-local-char}.


\section{Concluding Remarks}
When J\'anos Makowsky retired as President of the European Association for Computer Science Logic (EACSL), he gave an invited address at the 2011 CSL conference in which he reflected on his experience in research~\cite{DBLP:conf/csl/Makowsky11}. In particular, he reflected on his work in database theory, highlighted the undecidability of the implication problem for EIDs, and concluded the section on his contributions to database theory by writing that ``After that I tried to learn the true problems of database theory. However, J.\ Ullman changed his mind and declared that Dependency Theory and Design Theory had run their course. As a result, papers dealing with these topics were almost banned from the relevant conferences.'' Here, Makowsky refers to Ullman's invited talk and paper, titled ``Database Theory: Past and Future'', at the 1987 ACM Symposium on Principles of Database Systems (PODS)~\cite{DBLP:conf/pods/Ullman87}. When it came to the past of dependency theory, Ullman stated that ``We quickly learned far more about the subject than was necessary, a process that continues to this day'' and then, contemplating  the future, he delegated dependency theory to the section titled ``Last Gasps of the Dying Swans''. 

The pronouncement of the demise of dependency theory, however, turned out to be premature. Indeed, about a decade later, data dependencies found numerous uses in formalizing and analyzing data inter-operability tasks, such as 
data exchange and data integration (see, e.g., the surveys \cite{DBLP:conf/pods/Kolaitis05,DBLP:conf/pods/Lenzerini02} and the books \cite{DBLP:books/cu/ArenasBLM2014,DBLP:conf/dagstuhl/2013dfu5}). As a matter of fact, the study of data dependencies has enjoyed a renaissance that continues to date. In particular, tgds have been used to specify data transformation tasks (i.e., how data structured under one schema should be transformed into data structured under a different schema) and as building blocks of ontology languages~\cite{DBLP:journals/ai/BagetLMS11,DBLP:journals/ws/CaliGL12,DBLP:journals/ai/CaliGP12}. In the case of data exchange and data integration, a subclass of tgds, called \emph{source-to-target tgds} (s-t tgds), has played a central role. These are the tgds of the form $\forall \bar x\forall \bar y(\phi(\bar x,\bar y)\rightarrow \exists \bar z\psi(\bar x, \bar z))$, where all relations in $\phi(\bar x,\bar y)$ are from a source schema, while all relations in $\psi(\bar x, \bar z)$ are from a disjoint target schema. The study of s-t tgds led to the discovery of new structural properties for this class of tgds, the most prominent of which is the existence of \emph{universal} solutions in data exchange \cite{DBLP:journals/tcs/FaginKMP05}. In turn, this motivated the pursuit of characterizations of s-t tgds and of natural subclasses of them using structural properties relevant to data inter-operability, such as the existence of universal solutions and conjunctive query rewriting  \cite{DBLP:conf/icdt/CateK09,DBLP:journals/cacm/CateK10}. And now we have come full circle with the structural characterizations of arbitrary tgds and egds discussed in this paper. In conclusion, J\'anos Makowsky's work on data dependencies was well ahead of its time.

\bibliographystyle{plain}

\end{document}